\documentclass[12pt]{article}
\usepackage{amsthm}
\usepackage{amsmath}
\usepackage{amssymb}
\usepackage{upref}
\usepackage{amsfonts}

\newcommand{\cF}{{\cal F}}
\newcommand{\cC}{{\cal C}}

\newtheorem{theorem}{Theorem}
\newtheorem{lemma}{Lemma}

\newtheorem{corollary}{Corollary}
\newtheorem{example}{Example}

\begin{document}

\title{
Lower bounds on the minimum average distance of binary codes}

\author{
{\sc Beniamin Mounits}\thanks{CWI, Amsterdam, The Netherlands, e-mail: {\tt B.Mounits@cwi.nl}.}}

\maketitle

\begin{abstract}
Let $\beta(n,M)$ denote the minimum average Hamming distance of a binary code of length $n$ and cardinality $M.$ In this paper we consider lower bounds on $\beta(n,M).$ All the known lower bounds on $\beta(n,M)$ are useful when $M$ is at least of size about $2^{n-1}/n.$ We derive new lower bounds which give good estimations when size of $M$ is about $n.$ These bounds are obtained using linear programming approach. In particular, it is proved that $\displaystyle{\lim_{n\to\infty}\beta(n,2n)=5/2}.$ We also give new recursive inequality for $\beta(n,M).$
\end{abstract}

\vskip 0.1cm

\noindent {\bf Keywords:} Binary codes, minimum average distance, linear programming

%%%%%%%%%%%%%%%%%%%%%%%%%%%%%%%%%%%%%%%%%%%%%%%%%%%%%%%%%%%%%%%%%%%%%%
\newpage
\section{Introduction}

Let $\cF_2=\{0,1\}$ and let $\cF_2^n$ denotes the set of all binary words of length $n$. For $x,y\in \cF_2^n,$ $d(x,y)$ denotes the Hamming distance between $x$ and $y$ and $wt(x)=d(x,{\bf 0})$ is the weight of $x$, where ${\bf 0}$ denotes all-zeros word. A binary code $\cC$ of length $n$ is a nonempty subset of $\cF_2^n.$ An $(n,M)$ code $\cC$ is a binary code of length $n$ with cardinality $M.$ In this paper we will consider only binary codes.

The average Hamming distance of an $(n,M)$ code $\cC$ is defined by
\begin{align*}
\overline{d}(\cC)=\frac{1}{M^2}\sum_{c\in \cC}\sum_{c'\in \cC}d(c,c')~.
\end{align*}
The \emph{minimum average Hamming distance} of an $(n,M)$ code is defined by 
\begin{align*}
\beta(n,M)=\min\{~\overline{d}(\cC):~\cC~\textrm{is an}~(n,M)~\textrm{code}\}~.
\end{align*}
An $(n,M)$ code $\cC$ for which $\overline{d}(\cC)=\beta(n,M)$ will be called \emph{extremal} code.

The problem of determining $\beta(n,M)$ was proposed by Ahlswede and Katona in \cite{AhlKat}. Upper bounds on $\beta(n,M)$ are obtained by constructions. For survey on the known upper bounds the reader is referred to \cite{Kun}. In this paper we consider the lower bounds on $\beta(n,M).$
We only have to consider the case where $1\leq M\leq 2^{n-1}$ because of the following result which was proved in \cite{FuWeiYeu}.
\begin{lemma}\label{lem:complement_determ}
For $1\leq M\leq 2^n$
\begin{align*}
\beta(n,2^n-M)=\frac{n}{2}-\frac{M^2}{(2^n-M)^2}\left(\frac{n}{2}-\beta(n,M)\right)~.
\end{align*}
\end{lemma}

First exact values of $\beta(n,M)$ were found by Jaeger et al. \cite{JaeKheMol}.
\begin{theorem}\cite{JaeKheMol}
$\beta(n,4)=1,$ $\beta(n,8)=3/2,$ whereas for $M\leq n+1,$ $M\ne 4,8,$ we have $\displaystyle{\beta(n,M)=2\left(\frac{M-1}{M}\right)^2}.$
\end{theorem}

Next, Alth{\"{o}}fer and Sillke \cite{AltSill} gave the following bound.
\begin{theorem}\label{thm:Alt_Sil}\cite{AltSill}
\begin{align*}
\beta(n,M)\geq \frac{n+1}{2}-\frac{2^{n-1}}{M}~,
\end{align*}
where equality holds only for $M=2^n$ and $M=2^{n-1}.$
\end{theorem}
\noindent Xia and Fu \cite{XiaFu} improved Theorem \ref{thm:Alt_Sil} for odd $M.$
\begin{theorem}\label{thm:Xia_Fu}\cite{XiaFu}
If $M$ is odd, then
\begin{align*}
\beta(n,M)\geq \frac{n+1}{2}-\frac{2^{n-1}}{M}+\frac{2^n-n-1}{2M^2}~.
\end{align*}
\end{theorem}
\noindent Further, Fu et al. \cite{FuWeiYeu} found the following bounds.
\begin{theorem}\label{thm:Fu_Wei_Yeu}\cite{FuWeiYeu}
\begin{align*}
\beta(n,M)\geq\frac{n+1}{2}-\frac{2^{n-1}}{M}+\frac{2^n-2n}{M^2}~,~~\textrm{\emph{if}}~~M\equiv2(mod~4)~,\\
\\
\beta(n,M)\geq \frac{n}{2}-\frac{2^{n-2}}{M}~,~~\textrm{\emph{for}}~~M\leq 2^{n-1}~,~~~~~~~~~~~~~~~~~~~~~~\\
\\
\beta(n,M)\geq \frac{n}{2}-\frac{2^{n-2}}{M}+\frac{2^{n-1}-n}{2M^2}~,~~\textrm{\emph{if}}~~M~\textrm{\emph{is odd and}}~~M\leq 2^{n-1}-1~.
\end{align*}
\end{theorem}
\noindent Using Lemma \ref{lem:complement_determ} and Theorems \ref{thm:Xia_Fu}, \ref{thm:Fu_Wei_Yeu} the following values of $\beta(n,M)$ were determined:~\\ $\beta(n,2^{n-1}\pm1),$ $\beta(n,2^{n-1}\pm2),$ $\beta(n,2^{n-2}),$ $\beta(n,2^{n-2}\pm1),$ $\beta(n,2^{n-1}+2^{n-2}),$ $\beta(n,2^{n-1}+2^{n-2}\pm1).$
The bounds in Theorems \ref{thm:Xia_Fu}, \ref{thm:Fu_Wei_Yeu} were obtained by considering constraints on distance distribution of codes which were developed by Delsarte in \cite{Del1}. We will recall these constraints in the next section.

Notice that the previous bounds are only useful when $M$ is at least of size about $2^{n-1}/n.$ Ahlswede and Alth{\"{o}}fer determined $\beta(n,M)$ asymptotically.
\begin{theorem}\label{Ahl_Alt}\cite{AhlAlt}
Let $\{M_n\}_{n=1}^\infty$ be a sequence of natural numbers with $0\leq M_n\leq 2^n$ for all $n$ and $\displaystyle{\lim_{n\to \infty}\inf\left(M_n/\binom{n}{\lfloor \alpha n\rfloor}\right)>0}$ for some constant $\alpha,$ $0<\alpha<1/2.$ Then
\begin{align*}
\lim_{n\to \infty}\inf\frac{\beta(n,M_n)}{n}\geq 2\alpha(1-\alpha)~.
\end{align*}
\end{theorem}
\noindent The bound of Theorem \ref{Ahl_Alt} is asymptotically achieved by taking constant weight code~\\ $\cC=\{x\in \cF_2^n:~wt(x)=\lfloor \alpha n\rfloor\}.$

The rest of the paper is organized as follows. In Section \ref{Sec:Prelim} we give necessary background in linear programming approach for deriving bounds for codes. This includes Delsarte's inequalities on distance distribution of a code and some properties of binary Krawtchouk polynomials. In Section \ref{Sec:Large} we obtain lower bounds on $\beta(n,M)$ which are useful in case when $M$ is relatively large. In particular, we show that the bound of Theorem \ref{thm:Alt_Sil} is derived via linear programming technique. We also improve some bounds from Theorem \ref{thm:Fu_Wei_Yeu} for $M<2^{n-2}.$ In Section \ref{Sec:Small}, we obtain new lower bounds on $\beta(n,M)$ which are useful when $M$ is at least of size about $n/3.$ We also prove that these bounds are asymptotically tight for the case $M=2n.$ Finally, in Section \ref{Sec:Recursive}, we give new recursive inequality for $\beta(n,M).$

%%%%%%%%%%%%%%%%%%%%%%%%%%%%%%%%%%%%%%%%%%%%%%%%%%%%%%%%%%%%%%%%%%%%%%
\section{Preliminaries}\label{Sec:Prelim}

The distance distribution of an $(n,M)$ code $\cC$ is the $(n+1)$-tuple of rational numbers $\{A_0,A_1,\cdots,A_n\},$ where
\begin{align*}
A_i=\frac{|\{(c,c')\in \cC\times \cC:~d(c,c')=i\}|}{M}
\end{align*}
is the average number of codewords which are at distance $i$ from any given codeword $c\in \cC.$ It is clear that
\begin{align}\label{A_i_identities}
A_0=1~,~~\sum_{i=0}^nA_i=M~~\textrm{and}~~A_i\geq 0~~\textrm{for}~~0\leq i\leq n~.
\end{align}
If $\cC$ is an $(n,M)$ code with distance distribution $\{A_i\}_{i=0}^n,$ the dual distance distribution $\{B_i\}_{i=0}^n$ is defined by
\begin{align}\label{B_k_def}
B_k=\frac{1}{M}\sum_{i=0}^nP_k^n(i)A_i~,
\end{align}
where 
\begin{align}\label{Krawtchouk_def}
P_k^n(i)=\sum_{j=0}^k(-1)^j\binom{i}{j}\binom{n-i}{k-j}
\end{align}
is the binary Krawtchouk polynomial of degree $k$. It was proved by Delsarte \cite{Del1} that 
\begin{align}\label{Delsarte-Transf}
B_k\geq 0~~\textrm{for}~~0\leq k\leq n~.
\end{align}
Since the Krawtchouk polynomials satisfy the following orthogonal relation
\begin{align}\label{prop_5}
\sum_{k=0}^nP_k^n(i)P_j^n(k)=\delta_{ij}2^n~,
\end{align}
we have 
\begin{align}\label{dual_dual}
\sum_{k=0}^nP_j^n(k)B_k=\sum_{k=0}^nP_j^n(k)\frac{1}{M}\sum_{i=0}^nP_k^n(i)A_i=\frac{1}{M}\sum_{i=0}^nA_i\sum_{k=0}^nP_j^n(k)P_k^n(i)
=\frac{2^n}{M}A_j~.
\end{align}
It's easy to see from (\ref{A_i_identities}),(\ref{B_k_def}),(\ref{Krawtchouk_def}), and (\ref{dual_dual}) that
\begin{align}\label{B_i_identities}
B_0=1~~\textrm{and}~~\sum_{k=0}^nB_k=\frac{2^n}{M}~.
\end{align}

Before we proceed, we list some of the properties of binary Krawtchouk polynomials (see for example \cite{KrLi}).
\begin{itemize}
\item Some examples are: $P_0^n(x)\equiv 1,~P_1^n(x)=n-2x~,$
\begin{align*}
P_2^n(x)=\frac{(n-2x)^2-n}{2},~P_3^n(x)=\frac{(n-2x)((n-2x)^2-3n+2)}{6}~.
\end{align*}
\item For any polynomial $f(x)$ of degree $k$ there is the unique Krawtchouk expansion
\begin{align*}
f(x)=\sum_{i=0}^kf_iP_i^n(x)~,
\end{align*}
where the coefficients are
\begin{align*}
f_i=\frac{1}{2^n}\sum_{j=0}^nf(j)P_j^n(i)~.
\end{align*}
\item Krawtchouk polynomials satisfy the following recurrent relations:
\begin{align}\label{prop_1}
P_{k+1}^{n}(x)=\frac{(n-2x)P_{k}^{n}(x)-(n-k+1)P_{k-1}^{n}(x)}{k+1}~,
\end{align}
\begin{align}\label{prop_1-a}
P_{k}^{n}(x)=P_{k}^{n-1}(x)+P_{k-1}^{n-1}(x)~.
\end{align}
\item Let $i$ be nonnegative integer, $0\leq i\leq n.$ The following symmetry relations hold: 
\begin{align}\label{prop_2}
\binom{n}{i}P_k^n(i)=\binom{n}{k}P_i^n(k)~,
\end{align}
%\begin{align}\label{prop_3}
%P_k^n(i)=(-1)^kP_k^n(n-i)~,
%\end{align}
\begin{align}\label{prop_4}
P_k^n(i)=(-1)^iP_{n-k}^n(i)~.
\end{align}
\end{itemize}

%%%%%%%%%%%%%%%%%%%%%%%%%%%%%%%%%%%%%%%%%%%%%%%%%%%%%%%%%%%%%%%%%%%%%%
\section{Bounds for ``large'' codes}\label{Sec:Large}

The key observation for obtaining the bounds in Theorems \ref{thm:Xia_Fu}, \ref{thm:Fu_Wei_Yeu} is the following result.
\begin{lemma}\label{key_observ}\cite{XiaFu}
For an arbitrary $(n,M)$ code $\cC$ the following holds:
\begin{align*}
\overline{d}(\cC)=\frac{1}{2}\left(n-B_1\right)~.
\end{align*}
\end{lemma}
\noindent From Lemma \ref{key_observ} follows that any upper bound on $B_1$ will provide a lower bound on $\beta(n,M).$ We will obtain upper bounds on $B_1$ using linear programming technique.

Consider the following linear programming problem:
\vskip 3mm
~~~~~~~maximize~~$B_1$
\vskip 3mm
~~~~~~~subject to
\begin{align*}
\sum_{i=1}^nB_i=\frac{2^n}{M}-1~,
\end{align*}
\begin{align*}
\sum_{i=1}^nP_k^n(i)B_i\geq -P_k(0)~,~~1\leq k\leq n~,
\end{align*}

~~~~~~~and $B_i\geq 0$ for $1\leq i\leq n.$~\\
\noindent Note that the constraints are obtained from (\ref{dual_dual}) and (\ref{B_i_identities}).

The next theorem follows from the dual linear program. We will give an independent proof.
\newpage
\begin{theorem}\label{thm:large}
Let $\cC$ be an $(n,M)$ code such that for $2\leq i\leq n$ and $1\leq j\leq n$ there holds that $B_i\ne 0\Leftrightarrow i\in I$ and 
$A_j\ne 0\Leftrightarrow j\in J.$

Suppose a polynomial $\lambda(x)$ of degree at most $n$ can be found with the following properties. If the Krawtchouk expansion of $\lambda(x)$ is
\begin{align*}
\lambda(x)=\sum_{j=0}^n\lambda_jP_j^n(x)~,
\end{align*}
then $\lambda(x)$ should satisfy
\begin{align*}
\lambda(1)=-1~,~~~~~~\\
\lambda(i)\leq 0~~\textrm{for}~~i\in I~,~\\
\lambda_j\geq 0~~\textrm{for}~~j\in J~.~
\end{align*}
Then
\begin{align}\label{main_ineq_B_1}
B_1\leq \lambda(0)-\frac{2^n}{M}\lambda_0~.
\end{align}
The equality in (\ref{main_ineq_B_1}) holds iff $\lambda(i)=0$ for $i\in I$ and $\lambda_j=0$ for $j\in J.$
\end{theorem}
\begin{proof}
Let $\cC$ be an $(n,M)$ code which satisfies the above conditions. Thus, using (\ref{A_i_identities}), (\ref{B_k_def}), (\ref{Delsarte-Transf}) and (\ref{prop_5}), we have
\begin{align*}
-B_1=\lambda(1)B_1\geq \lambda(1)B_1+\sum_{i\in I}\lambda(i)B_i=\sum_{i=1}^n\lambda(i)B_i=\sum_{i=1}^n\lambda(i)\frac{1}{M}\sum_{j=0}^nP_i^n(j)A_j
\end{align*}
\begin{align*}
=\frac{1}{M}\sum_{j=0}^nA_j\sum_{i=1}^n\lambda(i)P_i^n(j)=\frac{1}{M}\sum_{j=0}^nA_j\sum_{i=1}^n\sum_{k=0}^n\lambda_kP_k^n(i)P_i^n(j)
\end{align*}
\begin{align*}
=\frac{1}{M}\sum_{j=0}^nA_j\sum_{k=0}^n\lambda_k\left(\sum_{i=0}^nP_k^n(i)P_i^n(j)-P_k^n(0)P_0^n(j)\right)=\frac{1}{M}\sum_{j=0}^nA_j\sum_{k=0}^n\lambda_k\delta_{kj}2^n
\end{align*}
\begin{align*}
-\frac{1}{M}\sum_{j=0}^nA_j\sum_{k=0}^n\lambda_kP_k^n(0)
=\frac{2^n}{M}\sum_{j=0}^n\lambda_jA_j-\lambda(0)=\frac{2^n}{M}\left(\lambda_0A_0+\sum_{j\in J}^n\lambda_jA_j\right)-\lambda(0)
\end{align*}
\begin{align*}
\geq \frac{2^n}{M}\lambda_0A_0-\lambda(0)=\frac{2^n}{M}\lambda_0-\lambda(0)~.
\end{align*}
\end{proof}

\newpage
\begin{corollary}\label{cor:first_bound}
If $\displaystyle{\lambda(x)=\sum_{j=0}^n\lambda_jP_j^n(x)}$ satisfies
\begin{enumerate}
\item $\lambda(1)=-1,$ $\lambda(i)\leq 0$ for $2\leq i\leq n,$
\item $\lambda_j\geq 0$ for $1\leq j\leq n,$
\end{enumerate}
then 
\begin{align*}
\beta(n,M)\geq \frac{1}{2}\left(n-\lambda(0)+\frac{2^n}{M}\lambda_0\right)~.
\end{align*}
\end{corollary}

\begin{example}
Consider the following polynomial:
\begin{align*}
\lambda(x)\equiv -1~.
\end{align*}
\end{example}
\noindent It is obvious that the conditions of the Corollary \ref{cor:first_bound} are satisfied. Thus we have a bound
\begin{align*}
\beta(n,M)\geq \frac{n+1}{2}-\frac{2^{n-1}}{M}
\end{align*}
which coincides with the one from Theorem \ref{thm:Alt_Sil}.

\begin{example}\label{exmple:2}\cite[Theorem 4]{FuWeiYeu}
Consider the following polynomial:
\begin{align*}
\lambda(x)=-\frac{1}{2}+\frac{1}{2}P_n^n(x)~.
\end{align*}
\end{example}
\noindent From (\ref{prop_4}) we see that
\begin{align*}
P_n^n(i)=(-1)^iP_0^n(i)=\left\{ \begin{array}{c}
1~~~~\textrm{if}~i~\textrm{is even}\\
-1~~~\textrm{if}~i~\textrm{is odd}~,
\end{array} \right.
\end{align*}
and, therefore,
\begin{align*}
\lambda(i)=\left\{ \begin{array}{c}
0~~~~\textrm{if}~i~\textrm{is even}\\
-1~~~\textrm{if}~i~\textrm{is odd}~.
\end{array} \right.
\end{align*}
Furthermore, $\lambda_j=0$ for $1\leq j\leq n-1$ and $\lambda_n=1/2.$ Thus, the conditions of the Corollary \ref{cor:first_bound} are satisfied and we obtain
\begin{align*}
\beta(n,M)\geq \frac{1}{2}\left(n-\frac{2^{n-1}}{M}\right)=\frac{n}{2}-\frac{2^{n-2}}{M}~.
\end{align*}
This bound was obtained in \cite[Theorem 4]{FuWeiYeu} and is tight for $M=2^{n-1},2^{n-2}.$

Other bounds in Theorems \ref{thm:Xia_Fu}, \ref{thm:Fu_Wei_Yeu} were obtained by considering additional constraints on distance distribution coefficients given in the next theorem.
\begin{theorem}\cite{BBMOS}
Let $\cC$ be an arbitrary binary $(n,M)$ code. If $M$ is odd, then
\begin{align*}
B_i\geq \frac{1}{M^2}\binom{n}{i}~,~~0\leq i\leq n~.
\end{align*}
If $M\equiv 2(mod~4),$ then there exists an $\ell\in \{0,1,\cdots,n\}$ such that
\begin{align*}
B_i\geq \frac{2}{M^2}\left(\binom{n}{i}+P_i^n(\ell)\right)~,~~0\leq i\leq n~.
\end{align*}
\end{theorem}

Next, we will improve the bound of Example \ref{exmple:2} for $M<2^{n-2}.$

\begin{theorem}
For $n>2$
\begin{align*}
\beta(n,M)\geq \left\{ \begin{array}{c}
\frac{n}{2}-\frac{2^{n-2}}{M}+\frac{1}{n-2}\left(\frac{2^{n-2}}{M}-1\right)~~~~\textrm{if}~n~\textrm{is even}\\
 \\
\frac{n}{2}-\frac{2^{n-2}}{M}+\frac{1}{n-1}\left(\frac{2^{n-2}}{M}-1\right)~~~~\textrm{if}~n~\textrm{is odd}~.
\end{array} \right.
\end{align*}
\end{theorem}

\begin{proof}
We distinguish between two cases.
\begin{itemize}
 \item If $n$ is even, $n>2,$ consider the following polynomial:
\begin{align*}
\lambda(x)=\frac{1}{2(n-2)}\left(3-n+P_{n-1}^n(x)+P_n^n(x)\right)~.
\end{align*}
Using (\ref{prop_4}), it's easy to see that
\begin{align*}
\lambda(i)=\left\{ \begin{array}{c}
\frac{2-i}{n-2}~~~~\textrm{if}~i~\textrm{is even}\\
\\
\frac{i+1-n}{n-2}~~~\textrm{if}~i~\textrm{is odd}~.
\end{array} \right.
\end{align*}
\item If $n$ is odd, $n>1,$ consider the following polynomial:
\begin{align*}
\lambda(x)=\frac{1}{2(n-1)}\left(2-n+P_{n-1}^n(x)+2P_n^n(x)\right)~.
\end{align*}
Using (\ref{prop_4}), it's easy to see that
\begin{align*}
\lambda(i)=\left\{ \begin{array}{c}
\frac{2-i}{n-1}~~~~\textrm{if}~i~\textrm{is even}\\
\\
\frac{i-n}{n-1}~~~\textrm{if}~i~\textrm{is odd}~.
\end{array} \right.
\end{align*}
\end{itemize}
\noindent In both cases, the claim of the theorem follows from Corollary \ref{cor:first_bound}.
\end{proof}

%%%%%%%%%%%%%%%%%%%%%%%%%%%%%%%%%%%%%%%%%%%%%%%%%%%%%%%%%%%%%%%%%%%%%%
\section{Bounds for ``small'' codes}\label{Sec:Small}

We will use the following lemma, whose proof easily follows from (\ref{prop_5}).
\begin{lemma}\label{lem:polynom_duality}
Let $\displaystyle{\lambda(x)=\sum_{i=0}^n\lambda_iP_i^n(x)}$ be an arbitrary polynomial. A polynomial~\\ $\displaystyle{\alpha(x)=\sum_{i=0}^n\alpha_iP_i^n(x)}$ satisfies $\alpha(j)=2^n\lambda_j$ iff $\alpha_i=\lambda(i).$
\end{lemma}

\noindent By substituting the polynomial $\lambda(x)$ from Theorem \ref{thm:large} into Lemma \ref{lem:polynom_duality}, we have the following.

\begin{theorem}\label{thm:small}
Let $\cC$ be an $(n,M)$ code such that for $1\leq i\leq n$ and $2\leq j\leq n$ there holds that $A_i\ne 0\Leftrightarrow i\in I$ and 
$B_j\ne 0\Leftrightarrow j\in J.$

Suppose a polynomial $\alpha(x)$ of degree at most $n$ can be found with the following properties. If the Krawtchouk expansion of $\alpha(x)$ is
\begin{align*}
\alpha(x)=\sum_{j=0}^n\alpha_jP_j^n(x)~,
\end{align*}
then $\alpha(x)$ should satisfy
\begin{align*}
\alpha_1=1~~,~~~~~~~~~~\\
\alpha_j\geq 0~~,~~\textrm{for}~~j\in J~,~\\
\alpha(i)\leq 0~~,~~\textrm{for}~~i\in I~.~~
\end{align*}
Then
\begin{align}\label{sec_main_ineq_B_1}
B_1\leq \frac{\alpha(0)}{M}-\alpha_0~.
\end{align}
The equality in (\ref{sec_main_ineq_B_1}) holds iff $\alpha(i)=0$ for $i\in I$ and $\alpha_j=0$ for $j\in J.$
\end{theorem}

\noindent Note that Theorem \ref{thm:small} follows from the dual linear program of the following one:
\vskip 3mm
~~~~~~~maximize~~$\displaystyle{\sum_{i=1}^nP_1^n(i)A_i=MB_1-n}$
\vskip 3mm
~~~~~~~subject to
\begin{align*}
\sum_{i=1}^nA_i=M-1~,
\end{align*}
\begin{align*}
\sum_{i=1}^nP_k^n(i)A_i\geq -P_k(0)~,~~1\leq k\leq n~,
\end{align*}

~~~~~~~and $A_i\geq 0$ for $1\leq i\leq n,$~\\
\noindent whose constraints are obtained from (\ref{A_i_identities}) and (\ref{Delsarte-Transf}).

\begin{corollary}\label{second_bound}
If $\displaystyle{\alpha(x)=\sum_{j=0}^n\alpha_jP_j^n(x)}$ satisfies
\begin{enumerate}
\item $\alpha_1=1,$ $\alpha_j\geq 0$ for $2\leq j\leq n,$
\item $\alpha(i)\leq 0$ for $1\leq i\leq n,$
\end{enumerate}
then 
\begin{align*}
\beta(n,M)\geq \frac{1}{2}\left(n+\alpha_0-\frac{\alpha(0)}{M}\right)~.
\end{align*}
\end{corollary}

\begin{example}
Consider
\begin{align*}
\alpha(x)=2-n+P_1^n(x)=2(1-x)~.
\end{align*}
\end{example}
\noindent It's obvious that the conditions of the Corollary \ref{second_bound} are satisfied and we obtain
\begin{theorem}\label{thm:1_zero}
\begin{align*}
\beta(n,M)\geq 1-\frac{1}{M}~.
\end{align*}
\end{theorem}
\noindent Note that the bound of Theorem \ref{thm:1_zero} is tight for $M=1,2.$

\begin{example}
Consider the following polynomial:
\begin{align*}
\alpha(x)=3-n+P_1^n(x)+P_n^n(x)~.
\end{align*}
\end{example}
\noindent From (\ref{prop_4}) we obtain
\begin{align*}
\alpha(i)=\left\{ \begin{array}{c}
4-2i~~~~\textrm{if}~i~\textrm{is even}\\
~2-2i~~~~\textrm{if}~i~\textrm{is odd}~.
\end{array} \right.
\end{align*}
\noindent Thus, conditions of the Corollary \ref{second_bound} are satisfied and we have
\begin{theorem}\label{thm:1_2_zeros}
\begin{align*}
\beta(n,M)\geq \frac{3}{2}-\frac{2}{M}~.
\end{align*}
\end{theorem}
\noindent Note that the bound of Theorem \ref{thm:1_2_zeros} is tight for $M=2,4.$

\begin{example}\label{exmpl:pol_1_2_3_zeros}
Let $n$ be even integer. Consider the following polynomial:
\begin{align}\label{pol_1_2_3_zeros}
\alpha(x)=\frac{n(4-n)}{n+2}+P_1^n(x)+\frac{4\binom{n}{2}}{(n+2)\binom{n}{\frac{n}{2}+1}} P_{\frac{n}{2}+1}^n(x)~.
\end{align}
\end{example}
\noindent In this polynomial $\alpha_1=1$ and $\alpha_j\geq 0$ for $2\leq j\leq n$. Thus, condition 1 in Corollary \ref{second_bound} is satisfied.
From (\ref{prop_2}) we obtain that for nonnegative integer $i,$ $0\leq i\leq n,$
\begin{align*}
P_{\frac{n}{2}+1}^n(i)=\frac{\binom{n}{\frac{n}{2}+1}}{\binom{n}{i}}P_i^n\left(\frac{n}{2}+1\right)
\end{align*}
and, therefore,
\begin{align}\label{pol_1_2_3_zeros-a}
\alpha(i)=\frac{n(4-n)}{n+2}+P_1^n(i)+\frac{4\binom{n}{2}}{(n+2)\binom{n}{i}} P_i^n\left(\frac{n}{2}+1\right)~.
\end{align}
It follows from (\ref{prop_1}) that
\begin{align*}
P_1^n\left(\frac{n}{2}+1\right)=-2~,~~P_2^n\left(\frac{n}{2}+1\right)=\frac{4-n}{2}~,~~P_3^n\left(\frac{n}{2}+1\right)=n-2~,
\end{align*}
\begin{align}\label{values_P_i_n/2+1}
P_4^n\left(\frac{n}{2}+1\right)=\frac{(n-2)(n-8)}{8}~,~~P_5^n\left(\frac{n}{2}+1\right)=\frac{(n-2)(4-n)}{4}~.
\end{align}
Now it's easy to verify from (\ref{pol_1_2_3_zeros-a}) and (\ref{values_P_i_n/2+1}) that $\alpha(1)=\alpha(2)=\alpha(3)=0.$ We define
\begin{align*}
\widetilde{\alpha}(i):=\frac{n(4-n)}{n+2}+P_1^n(i)+\frac{4\binom{n}{2}}{(n+2)\binom{n}{i}}\left|P_i^n\left(\frac{n}{2}+1\right)\right|~.
\end{align*}
It is clear that $\alpha(i)\leq \widetilde{\alpha}(i)$ for $0\leq i\leq n.$ We will prove that $\widetilde{\alpha}(i)\leq 0$ for $4\leq i\leq n.$
From (\ref{prop_4}) and (\ref{values_P_i_n/2+1}) one can verify that
\begin{align}\label{values_at_end}
\widetilde{\alpha}(n)=0~,~~\widetilde{\alpha}(n-1)=\widetilde{\alpha}(n-2)=\frac{2n(4-n)}{n+2}~,~~\textrm{and}~~\widetilde{\alpha}(n-3)=2(6-n)
\end{align}
which implies that $\widetilde{\alpha}(n-j)\leq 0$ for $0\leq j\leq 3$ (of course, we are not interested in values $\widetilde{\alpha}(n-j),$ $0\leq j\leq 3,$ if $n-j\in \{1,2,3\}$). So, it is left to prove that for every integer $i,$ $4\leq i\leq n-4,$ $\widetilde{\alpha}(i)\leq 0.$ Note that for an integer $i,$ $4\leq i\leq n/2,$
\begin{align*}
\widetilde{\alpha}(n-i)=\frac{n(4-n)}{n+2}+P_1^n(n-i)+\frac{4\binom{n}{2}}{(n+2)\binom{n}{n-i}}\left|P_{n-i}^n\left(\frac{n}{2}+1\right)\right|
\end{align*}
\begin{align*}
=\frac{n(4-n)}{n+2}+(2i-n)+\frac{4\binom{n}{2}}{(n+2)\binom{n}{i}}\left|(-1)^{\frac{n}{2}+1}P_{i}^n\left(\frac{n}{2}+1\right)\right|
\end{align*}
\begin{align*}
\leq \frac{n(4-n)}{n+2}+(n-2i)+\frac{4\binom{n}{2}}{(n+2)\binom{n}{i}}\left|P_{i}^n\left(\frac{n}{2}+1\right)\right|=\widetilde{\alpha}(i)~.
\end{align*}
Therefore, it is enough to check that $\widetilde{\alpha}(i)\leq 0$ only for $4\leq i\leq n/2.$

From (\ref{values_P_i_n/2+1}) we obtain that
\begin{align*}
\widetilde{\alpha}(4)=-2-\frac{6}{n-3}<0~~\textrm{and}~~\widetilde{\alpha}(5)=-4-\frac{12(n-8)}{(n+2)(n-3)}<0~,
\end{align*}
where, in view of (\ref{values_at_end}), we assume that $n\geq 8.$
To prove that $\widetilde{\alpha}(i)\leq 0$ for $6\leq i\leq n/2$ we will use the following lemma whose proof is given in the Appendix.
\begin{lemma}\label{lem:estimation}
If $n$ is an even positive integer and $i$ is an arbitrary integer number, $2\leq i\leq n/2,$ then
\begin{align*}
\left|P_i^n\left(\frac{n}{2}+1\right)\right|<\binom{n}{\lfloor \frac{i}{2}\rfloor}~.
\end{align*}
\end{lemma}
By Lemma \ref{lem:estimation}, the following holds for $2\leq i\leq n/2.$
\begin{align*}
\widetilde{\alpha}(i)=\frac{n(4-n)}{n+2}+P_1^n(i)+\frac{4\binom{n}{2}}{(n+2)\binom{n}{i}}\left|P_i^n\left(\frac{n}{2}+1\right)\right|
\end{align*}
\begin{align*}
<\frac{n(4-n)}{n+2}+n-2i+\frac{4\binom{n}{2}\binom{n}{\lfloor \frac{i}{2}\rfloor}}{(n+2)\binom{n}{i}}
=\frac{6n}{n+2}-2i+\frac{4\binom{n}{2}\binom{n}{\lfloor \frac{i}{2}\rfloor}}{(n+2)\binom{n}{i}}
\end{align*}
\begin{align*}
%=\frac{6n}{n+2}-6-2(i-3)+\frac{4\binom{n}{2}\binom{n}{\lfloor \frac{i}{2}\rfloor}}{(n+2)\binom{n}{i}}
=-\frac{12}{n+2}-2(i-3)+\frac{4\binom{n}{2}\binom{n}{\lfloor \frac{i}{2}\rfloor}}{(n+2)\binom{n}{i}}~.
\end{align*}
Thus, to prove that $\widetilde{\alpha}(i)\leq 0$ for $6\leq i\leq n/2,$ it's enough to prove that 
\begin{align*}
-2(i-3)+\frac{4\binom{n}{2}\binom{n}{\lfloor \frac{i}{2}\rfloor}}{(n+2)\binom{n}{i}}<0
\end{align*}
for $6\leq i\leq n/2.$

\begin{lemma}\label{lem:monot_sequence}
Let $n$ be an even integer. For $6\leq i\leq n/2$ we have
\begin{align*}
\frac{(i-3)\binom{n}{i}}{\binom{n}{\lfloor \frac{i}{2}\rfloor}}>\frac{n(n-1)}{n+2}~.
\end{align*}
\end{lemma}
\noindent The proof of this lemma appears in the Appendix.

We have proved that the both conditions of the Corollary \ref{second_bound} are satisfied and, therefore, for even integer $n,$ we have
\begin{align*}
\beta(n,M)\geq \frac{3n}{n+2}-\frac{n}{M}~.
\end{align*}

Once we have a bound for an even (odd) $n$, it's easy to deduce one for odd (even) $n$ due to the following fact which follows from (\ref{prop_1-a}).

\begin{lemma}\label{even-odd-relation}
Let $\displaystyle{\alpha(x)=\sum_{j=0}^n\alpha_jP_j^n(x)}$ be an arbitrary polynomial. Then for a polynomial
\begin{align*}
\mu(x)=\sum_{j=0}^{n-1}\mu_jP_j^{n-1}(x)~,
\end{align*}
where
\begin{align*}
\mu_j=\alpha_j+\alpha_{j+1}~,~~0\leq j\leq n-1~,
\end{align*}
the following holds:
\begin{align*}
\mu(x)=\alpha(x)~~\textrm{for}~~0\leq x\leq n-1~.
\end{align*}
\end{lemma}

\begin{example}\label{exmpl:n_odd_pol_1_2_3_zeros}
Let $n$ be odd integer, $n>1.$ Consider the following polynomial:
\begin{align}\label{n-odd-pol_1_2_3_zeros}
\mu(x)=\frac{6+3n-n^2}{n+3}+P_1^n(x)+\frac{4\binom{n+1}{2}}{(n+3)\binom{n+1}{\frac{n+3}{2}}}\left(P_{\frac{n+1}{2}}^n(x)
+P_{\frac{n+3}{2}}^n(x)\right)
\end{align}
\end{example}
\noindent which is obtained from $\alpha(x)$ given in (\ref{pol_1_2_3_zeros}) by the construction of Lemma \ref{even-odd-relation}. Thus, by Corollary \ref{second_bound}, for odd integer $n,$ we have
\begin{align*}
\beta(n,M)\geq \frac{3(n+1)}{n+3}-\frac{n+1}{M}~.
\end{align*}

We summarize the bounds from the Examples \ref{exmpl:pol_1_2_3_zeros}, \ref{exmpl:n_odd_pol_1_2_3_zeros} in the next theorem.
\begin{theorem}\label{thm:1_2_3_zeros}
\begin{align*}
\beta(n,M)\geq \left\{ \begin{array}{c}
\frac{3n}{n+2}-\frac{n}{M}~~~~~~~\textrm{if}~n~\textrm{is even}\\
 \\
\frac{3(n+1)}{n+3}-\frac{n+1}{M}~~~~\textrm{if}~n~\textrm{is odd}~.
\end{array} \right.
\end{align*}
\end{theorem}

\begin{example}
For $n\equiv1~(mod~4),~n\ne 1,$ consider
\begin{align}\label{odd-pol_1_2_3_4_zeros}
\alpha(x)=\frac{(1-n)(n-5)}{n+1}+P_1^n(x)+\frac{4n(n-2)}{(n+1)\binom{n}{\frac{n+1}{2}}} P_{\frac{n+1}{2}}^n(x)+P_n^n(x)~.
\end{align}
\end{example}
\noindent One can verify that
\begin{align*}
\alpha(0)=4(n-1)~,~~\alpha(1)=\alpha(2)=\alpha(3)=\alpha(4)=0~,~~\alpha(5)=\alpha(6)=\frac{4(1-n)}{n-4}~,
\end{align*}
and
\begin{align*}
\alpha(n)=-6\frac{(n-1)^2}{n+1}~,~~\alpha(n-1)=\alpha(n-2)=\alpha(n-3)=\alpha(n-4)=-2\frac{(n-5)(n-1)}{n+1}~,~\\
\alpha(n-5)=\alpha(n-6)=-\frac{2(n-9)(n-2)(n-1)}{(n+1)(n-4)}~.~~~~~~~~~~~~~~~~~~~~~~~~~~
\end{align*}
We define
\begin{align*}
\widetilde{\alpha}(i):=\frac{(1-n)(n-5)}{n+1}+P_1^n(x)+\frac{4n(n-2)}{(n+1)\binom{n}{i}} \left|P_i^n\left(\frac{n+1}{2}\right)\right|+\left|P_n^n(i)\right|~.
\end{align*}
As in the previous example, it's easy to see that $\alpha(i)\leq \widetilde{\alpha}(i)$ for $0\leq i\leq n$ and 
\begin{align*}
\widetilde{\alpha}(n-i)\leq \widetilde{\alpha}(i)~~\textrm{for}~~0\leq i\leq (n-1)/2~.
\end{align*}
Therefore, to prove that $\alpha(i)\leq 0$ for $1\leq i\leq n,$ we only have to show that $\widetilde{\alpha}(i)\leq 0$ for $7\leq i\leq (n-1)/2.$ It is follows from the next two lemmas.

\begin{lemma}\label{lem:sec_estimation}
If $n$ is odd positive integer and $i$ is an arbitrary integer number, $2\leq i\leq (n-1)/2,$ then
\begin{align*}
\left|P_i^n\left(\frac{n+1}{2}\right)\right|<\binom{n}{\lfloor \frac{i}{2}\rfloor}~.
\end{align*}
\end{lemma}

\begin{lemma}\label{lem:sec_monot_sequence}
Let $n$ be odd integer. For $7\leq i\leq (n-1)/2$ we have
\begin{align*}
\frac{(i-4)\binom{n}{i}}{\binom{n}{\lfloor \frac{i}{2}\rfloor}}>\frac{2n(n-2)}{n+1}~.
\end{align*}
\end{lemma}
\noindent Proofs of the Lemmas \ref{lem:sec_estimation}, \ref{lem:sec_monot_sequence} are very similar to those of Lemmas \ref{lem:estimation}, \ref{lem:monot_sequence}, respectively, and they are omitted. Thus, we have proved that the conditions of the Corollary \ref{second_bound} are satisfied and we have the following bound.
\begin{align*}
\beta(n,M)\geq \frac{7n-5}{2(n+1)}-\frac{2(n-1)}{M}~,~~\textrm{if}~n\equiv1~(mod~4)~,~~n\ne 1~.
\end{align*}
From Lemma \ref{even-odd-relation}, by choosing the following polynomials:
\begin{align*}
\mu(x)=\frac{2+5n-n^2}{n+2}+P_1^n(x)+\frac{4(n^2-1)}{(n+2)\binom{n+1}{\frac{n+2}{2}}}
\left(P_{\frac{n}{2}}^n(x)+P_{\frac{n+2}{2}}^n(x)\right)+ P_n^n(x)~,
\end{align*}
if $n\equiv0~(mod~4),$
\begin{align*}
\widetilde{\mu}(x)=\frac{9+4n-n^2}{n+3}+P_1^n(x)+\frac{4n(n+2)}{(n+3)\binom{n+2}{\frac{n+3}{2}}} \left(P_{\frac{n-1}{2}}^n(x)+P_{\frac{n+3}{2}}^n(x)\right)~\\
+\frac{8n(n+2)}{(n+3)\binom{n+2}{\frac{n+3}{2}}} P_{\frac{n+1}{2}}^n(x)+ P_n^n(x)~,~~~~~~~~~~~~~~~~~~~
\end{align*}
if $n\equiv3~(mod~4),~n\ne 3,$ and 
\begin{align*}
\widehat{\mu}(x)=\frac{16+3n-n^2}{n+4}+P_1^n(x)+\frac{4(n+1)(n+3)}{(n+4)\binom{n+3}{\frac{n+4}{2}}} \left(P_{\frac{n-2}{2}}^n(x)+P_{\frac{n+4}{2}}^n(x)\right)~\\
+\frac{12(n+1)(n+3)}{(n+4)\binom{n+3}{\frac{n+4}{2}}} \left(P_{\frac{n}{2}}^n(x)+P_{\frac{n+2}{2}}^n(x)\right)+ P_n^n(x)~,~~~~~~~~~~~~~~
\end{align*}
if $n\equiv2~(mod~4),~n\ne 2,$ we obtain the bounds which are summarized in the next theorem.

\begin{theorem}\label{thm:1_2_3_4_zeros}
For $n>3$
\begin{align*}
\beta(n,M)\geq \left\{ \begin{array}{c}
\frac{7n+2}{2(n+2)}-\frac{2n}{M}~~~~~~~~~\textrm{if}~n\equiv0~(mod~4)~\\
 \\
\frac{7n-5}{2(n+1)}-\frac{2(n-1)}{M}~~~~\textrm{if}~n\equiv1~(mod~4)~\\
 \\
\frac{7n+16}{2(n+4)}-\frac{2(n+2)}{M}~~~~\textrm{if}~n\equiv2~(mod~4)~\\
 \\
\frac{7n+9}{2(n+3)}-\frac{2(n+1)}{M}~~~~\textrm{if}~n\equiv3~(mod~4)~.
\end{array} \right.
\end{align*}
\end{theorem}

It's easy to see that the bounds of Theorems \ref{thm:1_2_3_zeros} and \ref{thm:1_2_3_4_zeros} give similar estimations when the size of a code is about $2n.$

\begin{theorem}
\begin{align*}
\lim_{n\to \infty}\beta(n,2n)=\frac{5}{2}~.
\end{align*}
\end{theorem}

\begin{proof}
Let $\cC$ be the following $(n,2n)$ code:
\begin{align*}
\left. \begin{array}{ccc}
000 & \cdots & 00 \\
\hline
100 & \cdots & 00 \\
010 & \cdots & 00 \\
\vdots & \ddots & \vdots \\
000 & \cdots & 01 \\
\hline
110 & \cdots & 00 \\
101 & \cdots & 00 \\
\vdots & \ddots & \vdots \\
100 & \cdots & 01 \\
\end{array} \right.
\end{align*}
One can evaluate that 
\begin{align}\label{upper_estim}
\beta(n,2n)\leq \overline{d}(\cC)=\frac{5}{2}-\frac{4n-2}{n^2}~.
\end{align}
On the other hand, Theorem \ref{thm:1_2_3_zeros} gives
\begin{align}\label{lower_estim}
\beta(n,2n)\geq \left\{ \begin{array}{c}
\frac{5}{2}-\frac{6}{n+2}~~~~~~~\textrm{if}~n~\textrm{is even}\\
 \\
\frac{5}{2}-\frac{13n+3}{2n(n+3)}~~~~\textrm{if}~n~\textrm{is odd}~.
\end{array} \right.
\end{align}
The claim of the theorem follows by combining (\ref{upper_estim}) and (\ref{lower_estim}).
\end{proof}

%%%%%%%%%%%%%%%%%%%%%%%%%%%%%%%%%%%%%%%%%%%%%%%%%%%%%%%%%%%%%%%%%%%%%%
\section{Recursive inequality on $\beta(n,M)$}\label{Sec:Recursive}

The following recursive inequality was obtained in \cite{XiaFu}:

\begin{align}\label{Xia_Fu_recursive}
\beta(n,M+1)\geq \frac{M^2}{(M+1)^2}\beta(n,M)+\frac{Mn}{(M+1)^2}\left(1-\sqrt{1-\frac{2}{n}\beta(n,M)}\right)~.
\end{align}

\noindent In the next theorem we give a new recursive inequality.
\begin{theorem}\label{thm:monotonicity}
For positive integers $n$ and $M,$ $2\leq M\leq 2^n-1,$
\begin{align}\label{eqtn:monotonicity}
\beta(n,M+1)\geq \frac{M^2}{M^2-1}\beta(n,M)~.
\end{align}
\end{theorem}

\begin{proof}
Let $\cC$ be an extremal $(n,M+1)$ code, i.e.,
\begin{align*}
\beta(n,M+1)=\overline{d}(\cC)=\frac{1}{(M+1)^2}\sum_{c\in \cC}\sum_{c'\in \cC}d(c,c')~.
\end{align*}
Then there exists $c_0\in \cC$ such that 
\begin{align}\label{eqtn:averaging}
\sum_{c\in \cC}d(c_0,c)\geq (M+1)\beta(n,M+1)~.
\end{align}
Consider an $(n,M)$ code $\widetilde{\cC}=\cC\setminus \{c_0\}.$ Using (\ref{eqtn:averaging}) we obtain
\begin{align*}
\beta(n,M)\leq \overline{d}(\widetilde{\cC})=\frac{1}{M^2}\sum_{c\in \widetilde{\cC}}\sum_{c'\in \widetilde{\cC}}d(c,c')
=\frac{1}{M^2}\left(\sum_{c\in \cC}\sum_{c'\in \cC}d(c,c')-2\sum_{c\in \cC}d(c_0,c)\right)
\end{align*}
\begin{align*}
\leq \frac{1}{M^2}\left((M+1)^2\beta(n,M+1)-2(M+1)\beta(n,M+1)\right)=\frac{M^2-1}{M^2}\beta(n,M+1)~.
\end{align*}
\end{proof}

\begin{lemma}
For positive integers $n$ and $M,$ $2\leq M\leq 2^n-1,$ the RHS of (\ref{eqtn:monotonicity}) is not smaller than RHS of (\ref{Xia_Fu_recursive}).
\end{lemma}

\begin{proof}
One can verify that RHS of (\ref{eqtn:monotonicity}) is not smaller than RHS of (\ref{Xia_Fu_recursive}) iff
\begin{align*}
\beta(n,M)\leq \frac{M^2-1}{M^2}\cdot \frac{n}{2}~.
\end{align*}
By (\ref{eqtn:monotonicity}) we have
\begin{align*}
\beta(n,M)\leq \frac{M^2-1}{M^2}\beta(n,M+1)\leq \frac{M^2-1}{M^2}\beta(n,2^n)=\frac{M^2-1}{M^2}\cdot \frac{n}{2}~,
\end{align*}
which completes the proof.
\end{proof}

%%%%%%%%%%%%%%%%%%%%%%%%%%%%%%%%%%%%%%%%%%%%%%%%%%%%%%%%%%%%%%%%%%%%%%
\section{Appendix}

{\bf Proof of Lemma \ref{lem:estimation}:} The proof is by induction. One can easily see from (\ref{values_P_i_n/2+1}) that the claim is true for $2\leq i\leq 5,$ where $i\leq n/2.$ Assume that we have proved the claim for $i,$ $4\leq i\leq k\leq n/2-1.$ Thus
\begin{align*}
\left|P_{k+1}^n\left(\frac{n}{2}+1\right)\right|=\left|\frac{(-2)P_{k}^{n}\left(\frac{n}{2}+1\right)-(n-k+1)
P_{k-1}^{n}\left(\frac{n}{2}+1\right)}{k+1}\right|
\end{align*}
\begin{align*}
\leq \frac{2}{k+1}\left|P_{k}^{n}\left(\frac{n}{2}+1\right)\right|+\frac{n-k+1}{k+1}\left|P_{k-1}^{n}\left(\frac{n}{2}+1\right)\right|
\end{align*}
\begin{align*}
< \frac{2}{k+1}\binom{n}{\lfloor \frac{k}{2}\rfloor}+\frac{n-k+1}{k+1}\binom{n}{\lfloor \frac{k-1}{2}\rfloor}=(*)~.
\end{align*}
We distinguish between two cases. If $k$ is odd, then
\begin{align*}
(*)=\frac{2}{k+1}\binom{n}{\frac{k-1}{2}}+\frac{n-k+1}{k+1}\binom{n}{\frac{k-1}{2}}
=\frac{2}{k+1}\binom{n}{\frac{k-1}{2}}\left(1+\frac{n-k+1}{2}\right)
\end{align*}
\begin{align*}
=\frac{1}{n-\frac{k-1}{2}}\cdot \frac{n-\frac{k-1}{2}}{\frac{k+1}{2}}\binom{n}{\frac{k-1}{2}}\frac{n-k+3}{2}
=\frac{n-k+3}{2n-k+1}\binom{n}{\frac{k+1}{2}}<\binom{n}{\frac{k+1}{2}}~.
\end{align*}
Therefore, for odd $k,$ we obtain
\begin{align*}
\left|P_{k+1}\left(\frac{n}{2}+1\right)\right|<\binom{n}{\frac{k+1}{2}}=\binom{n}{\lfloor \frac{k+1}{2}\rfloor}~.
\end{align*}
If $k$ is even, then
\begin{align*}
(*)=\frac{2}{k+1}\binom{n}{\frac{k}{2}}+\frac{n-k+1}{k+1}\binom{n}{\frac{k}{2}-1}
\end{align*}
\begin{align*}
=\frac{2}{k+1}\binom{n}{\frac{k}{2}}+\frac{n-k+1}{k+1}\cdot \frac{\frac{k}{2}}{n-(\frac{k}{2}-1)}\cdot 
\frac{n-(\frac{k}{2}-1)}{\frac{k}{2}}\binom{n}{\frac{k}{2}-1}
\end{align*}
\begin{align*}
=\binom{n}{\frac{k}{2}}\left(\frac{2}{k+1}+\frac{n-k+1}{2n-k+2}\cdot \frac{k}{k+1}\right)~.
\end{align*}
Since $k\geq 4,$ we have
\begin{align*}
(*)=\binom{n}{\frac{k}{2}}\left(\frac{2}{k+1}+\overbrace{\frac{n-k+1}{2n-k+2}}^{<1/2}\cdot \overbrace{\frac{k}{k+1}}^{<1}\right)
<\binom{n}{\frac{k}{2}}\left(\frac{2}{5}+\frac{1}{2}\right)<\binom{n}{\frac{k}{2}}~.
\end{align*}
Therefore, for even $k,$ we obtain
\begin{align*}
\left|P_{k+1}\left(\frac{n}{2}+1\right)\right|<\binom{n}{\frac{k}{2}}=\binom{n}{\lfloor \frac{k+1}{2}\rfloor}~.
\end{align*}
\qed

\noindent {\bf Proof of Lemma \ref{lem:monot_sequence}:} Denote
\begin{align*}
a_i=\frac{(i-3)\binom{n}{i}}{\binom{n}{\lfloor \frac{i}{2}\rfloor}}~,~~6\leq i\leq n/2~.
\end{align*}
Thus,
\begin{align*}
\frac{a_6(n+2)}{n(n-1)}=\frac{(n+2)(n-3)(n-4)(n-5)}{40n(n-1)}
\end{align*}
\begin{align*}
=\frac{(n-2)(n-7)}{40}+\frac{48n-120}{40n(n-1)}\overbrace{\geq}^{n\geq 12}\frac{5}{4}+\frac{48\cdot 12-120}{40n(n-1)}>\frac{5}{4}
\end{align*}
and we have proved that $\displaystyle{a_6>\frac{n(n-1)}{n+2}}.$ Let's see that $a_i\geq a_6$ for $6\leq i\leq n/2.$ Let $i$ be even integer such that $6\leq i\leq n/2-2.$ Then
\begin{align*}
\frac{a_{i+2}}{a_i}=\frac{(i-1)(n-i-1)(n-i)}{(i-3)(i+1)(n-2i)}\overbrace{>}^{i\geq 6} \frac{(i-3)(n-2i)(n-i)}{(i-3)(i+1)(n-2i)}=\frac{n-i}{i+1}\overbrace{>}^{i\leq n/2-2}1~.
\end{align*}
Together with $\displaystyle{a_6>\frac{n(n-1)}{n+2}},$ this implies that $\displaystyle{a_i>\frac{n(n-1)}{n+2}}$ for every even integer $i,$~\\ $6\leq i\leq n/2.$

Now let $i$ be even integer such that $6\leq i\leq n/2-1.$ Then
\begin{align*}
\frac{a_{i+1}}{a_i}=\frac{(i-2)(n-i)}{(i-3)(i+1)}>\frac{n-i}{i+1}\overbrace{>}^{i\leq n/2-1}1~,
\end{align*}
which completes the proof.
\qed

%\section*{~~~~~~~~~~~~~~~~~~~~~~~Acknowledgment}

%The authors wish to thank two anonymous referees for their constructive comments.

\end{document}